\documentclass[11p,reqno]{amsart}

\usepackage[T1]{fontenc}
\usepackage{graphicx}
\usepackage{amssymb,amsthm,amsmath,mathrsfs,bm,braket,marginnote}
\usepackage{enumerate}
\usepackage{appendix}
\usepackage[colorlinks=true, pdfstartview=FitV, linkcolor=blue, citecolor=blue, urlcolor=blue]{hyperref}
\usepackage{multirow}

\usepackage{pgf}
\usepackage{pgfplots}
\usepackage{tikz}
\usetikzlibrary{arrows,calc}
\usepackage{verbatim}
\usetikzlibrary{decorations.pathreplacing,decorations.pathmorphing}
\usepackage[numbers,sort&compress]{natbib}
\usepackage{dsfont}

\numberwithin{equation}{section}
\newtheorem{theorem}{Theorem}[section]

\newtheorem{remark}[theorem]{Remark}

\newtheorem{proposition}[theorem]{Proposition}

 \reversemarginpar
 \newcommand{\pf}{\text{pf}}
 \newcommand{\Pf}{\text{Pf}}
 \newcommand{\md}{\mathcal{D}}

\newcommand{\p}{\partial}

\begin{document}

\title[Discrete integrable systems and condensation algorithms for Pfaffians]{Discrete integrable systems and condensation algorithms for Pfaffians}

\subjclass[2010]{{37K10, 15A15, 65D15}}
\date{}

\dedicatory{}
\author{Shi-Hao Li}
\address{ School of Mathematical and Statistics, ARC Centre of Excellence for Mathematical and Statistical Frontiers, The University of Melbourne, Victoria 3010, Australia}
\email{lishihao@lsec.cc.ac.cn}

\keywords{Discrete integrable systems,\, Pfaffian tau functions,\, condensation algorithm}
\begin{abstract}
Inspired by the connection between the Dodgson's condensation algorithm and Hirota's difference equation, we consider  condensation algorithms for Pfaffians from the perspectives of discrete integrable systems. The discretisation of Pfaffian elements demonstrate its effectiveness to the Pfaffian $\tau$-functions and discrete integrable systems. The free parameter in the discretisation allows us in particular to obtain explicit, one-parameter condensation algorithms for the Pfaffians.
\end{abstract}

\maketitle

\section{Introduction and statement of results}
Integrable combinatorics is a fascinating subject since it collects the ideas from integrable system, combinatorics, cluster algebra and so on. One of the main topic is the Dodgson's condensation algorithm proposed by Dodgson \cite{dodgson66},
based on the celebrated Jacobi identity
\begin{align*}
|A|\times |A^{1,n}_{1,n}|=|A^{1}_1|\times |A_n^n|-|A_n^1|\times |A_1^n|,
\end{align*}
where $|A_{i_1,\cdots,i_r}^{j_1,\cdots,j_r}|$ stands for the determinant of the the matrix obtained from $A$ by deleting its $i_1,\cdots,i_r$ rows and $j_1,\cdots,j_r$ columns. According to the Jacobi identity, choosing $\tau^{k,l}_n$ as an $n$-th order determinant
$\det(a_{i,j})$, where $i$ is from $k$ to $k+n-1$ and $j$ is from $l$ to $l+n-1$, specifies
the famous discrete Toda equation (or Hirota-Miwa equation, discrete KP equation (dKP) and $A_\infty$ T-system)
\begin{align} \label{toda}
\tau^{k,l}_{n+1}\tau^{k+1,l+1}_{n-1}=\tau_n^{k,l}\tau^{k+1,l+1}_n-\tau_n^{k+1,l}\tau^{k,l+1}_n.
\end{align}
One of the most important features of this integrable lattice is its exact solvability. That is, if we set enough initial values $\tau^{k,l}_0=1$ and $\tau_1^{k,l}=a_{k,l}$ for all $k,l\in\{0,\cdots,N-1\}$, then by iterating the equation \eqref{toda}, we can get the value of $\tau_{N}^{0,0}$, which is equal to the  $N$-th order determinant $\det(a_{k,l})_{k,l=0}^{N-1}$ (See Fig \ref{p} for a one-step process).
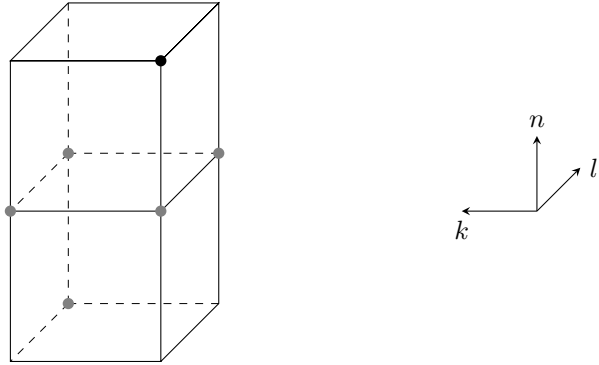
\begin{figure}[htbp]
\begin{center}
\begin{tikzpicture}
\newcommand*\kante{2}
\path[scale=\kante]
  (0,0,0)  coordinate   (A)
  (1,0,0)  coordinate   (B)
  (1,1,0)  coordinate   (C) 
  (0,1,0)  coordinate   (D)
  (0,0,-1) coordinate   (E)
  (1,0,-1) coordinate  (F)
  (1,1,-1) coordinate  (G)
  (0,1,-1) coordinate   (H)
   (0,2,0)  coordinate (A1)
  (1,2,0)  coordinate  (B1)
  (0,2,-1) coordinate   (E1)
  (1,2,-1) coordinate  (F1)

  (3.5,1,0)     coordinate (O)
  (3.5,1.5,0)  coordinate [label=above:$n$] (N)
  (3.5,1,-.75) coordinate [label=right:$l$](L)
  (3,1,0)  coordinate [label=below:$k$] (K)
  
;
\draw[font=\tiny]
(A) -- (A1) -- (B1) -- (B) -- cycle

(B) -- (F) -- (F1) -- (B1)

(D) -- (C) -- (G)

(A1) -- (B1) -- (F1)--(E1)--cycle
;

\path[dashed, very thin] (E) edge (A) edge (F) edge (H);
\path[dashed, very thin] (H) edge (D) edge (G) edge (E1);

\foreach \n in {D,H,C,G,E}
  \node at (\n)[circle,fill=gray,inner sep=1.5pt]{};
  
\foreach \l in{B1}
  \node at (\l)[circle,fill,inner sep=1.5pt]{};

\draw [->,>=stealth] (O) -- (N) ;
\draw [->,>=stealth] (O) -- (L);
\draw [->,>=stealth] (O) -- (K);

\end{tikzpicture}\caption{An graphic explanation of the Toda lattice: if one gets enough values at the bottom and middle levels, then one can get the value at the top by using the iterating formula.}\label{p}
\end{center}
\end{figure}

Later, the $\lambda$-determinant was introduced \cite{robbins86} as a generalisation of the ordinary determinant via a one-parameter Dodgson's condensation algorithm
\begin{align*}
T^{k,l}_{n+1}T^{k,l}_{n-1}=T^{k,l+1}_nT^{k,l-1}_n+\lambda T^{k+1,l}_nT^{k-1,l}_n,
\end{align*}
where $T^{k,l}_n$ is connected with $\tau^{k,l}_n$ via the affine transformation $\tau^{k,l}_{n}\mapsto T^{k-l,k+l+n}_n$.
The iteration process produces the corresponding $\lambda$-determinant. If one considers initial data \cite{difrancesco12}
\begin{align*}
&T^{l,m}_{0}=1, \qquad\qquad\qquad\qquad\qquad\,\, \, (l,m\in\mathbb{Z};\, l+m=n\,\text{mod}\,2)\\
&T^{i,j}_{1}=a_{(j-i+n+1)/2,(i+j+n+1)/2},\quad (i,j\in\mathbb{Z};\, i+j=n+1\,\text{mod}\, 2; \,|i|+|j|\leq n-1)
\end{align*}
then the solution $T_{n}^{0,0}:=|A|_\lambda$ is called as the $\lambda$-determinant. Moreover, the Laurent polynomial induced by the $\lambda$-determinant is related to the so-called alternating sign matrices, which is useful in the six-vertex model and thus attracted much attention from combinatorics and mathematical physics \cite{difrancesco12,kuo06,hone07}.

The cube recurrence
\begin{align}\label{btoda1}
\tau_n^{k+1,l+1}\tau_n^{k,l}-\tau_n^{k,l+1}\tau_n^{k+1,l}=\tau_{n-1}^{k+1,l+1}\tau_{n+1}^{k,l}-\tau_{n+1}^{k+1,l}\tau_{n-1}^{k,l+1},
\end{align}
is another discrete dynamical system admitting Laurent property \cite{fomin02} and it has two different guises. One was given by Miwa \cite{miwa82} as an analogy of Hirota's difference equation. Later, the same recurrence relation was introduced by Propp \cite{propp01} and studied from the point of view of algebraic combinatorics \cite{carroll04,fomin02}.

From the perspectives of integrable system, one can always expect that the $\tau$-functions of the $B$-type lattice could be expressed as Pfaffians \cite{hirota04} and it is natural to ask whether one can compute the value of Pfaffian via the recurrence relation, i.e. to find a condensation algorithm for Pfaffian. Although we will show that equation \eqref{btoda1} has Pfaffian solutions, the scheme is not explicit for iteration. Therefore, in this text, we mainly investigate the condensation algorithms for Pfaffians. For this purpose, we introduce 
two different integrable lattices to compute the value of Pfaffians. One is the B\"acklund transformation of the cube recurrence/dBKP equation, which rotates the vertices and make it explicit to iterate. Another is a so-called D-type Toda lattice which was proposed when investigating a convergent acceleration algorithm for Pfaffian sequence transformations. One-parameter recurrence relations are also proposed. With proper discrete time evolutions, we show that these integrable lattices have closed form solutions with a parameter, which was called a relaxation factor in the algorithm.

We organise the article as following. In Section \ref{intr}, we give some brief introductions to Pfaffian for later use.
Next, the condensation algorithms are given in Section \ref{ca}. Since the discrete evolutions of the solutions admit a one-parameter deformation, we introduce relaxation factors in the algorithm. As a by-product, in Section \ref{dc}, we give an explicit Pfaffian solution to the discrete CKP equation, demonstrating the effectiveness of this Pfaffian technique, which is powerful in discretisation of Pfaffian $\tau$-function as well as integrable systems. Some concluding remarks are given at the end.

\section{A gentle introduction to Pfaffian}\label{intr}
Given a skew-symmetric matrix of order $2N$, $A:=(a_{i,j})_{i,j=1}^{2N}$, then the Pfaffian of $A$ is given by 
\begin{align*}
\Pf(A)=\sum_P (-1)^P a_{i_1,i_2}\cdots a_{i_{2N-1},i_{2N}}.
\end{align*}
The summation means the sum over all possible combinations of pairs selected from $1$ to $2N$ satisfying
$i_{2l-1}<i_{2l+1}$ and $i_{2l-1}<i_{2l}$. The factor $(-1)^P$ takes the value $+1$ ($-1$) if the sequence $i_1$ to $i_{2N}$ is an even (odd) permutation of $1$ to $2N$. Here we take the notation 
\begin{align*}
\Pf(A):=\Pf(1,\cdots,2N)=\sum_P (-1)^P \Pf(i_1,i_2)\cdots \Pf(i_{2N-1},i_{2N})
\end{align*}
with $\Pf(i,j)=a_{i,j}$. From the definition, one can find the following explicit expansion
\begin{align}\label{exp}
\Pf(1,\cdots,2N)=\sum_{j=2}^{2N}(-1)^{j}\Pf(1,j)\Pf(2,\cdots,\hat{j},\cdots,2N),
\end{align}
where $\hat{j}$ denotes that the index $j$ is omitted. If one computes the value of Pfaffian by directly using \eqref{exp}, then the cost of this algorithm is about $O(N!)$ floating-point operation. Therefore, there have been some efficient numerical algorithms for computing the Pfaffians \cite{wimmer12,rubow11} whose computational cost is of $O(N^3)$. One idea is to find the canonical form by making the use of skew LU decomposition given by \cite{bunch82}: it is well known that for a $2n\times 2n$ skew-symmetric matrix $A$ and an arbitrary $2n\times 2n$ matrix $B$, there holds
\begin{align}\label{trans}
\Pf(BAB^{T})=\det(B)\Pf(A).
\end{align}
Therefore, to compute the Pfaffian of $A$ is to compute the value of a $2\times 2$ block diagonal matrix after the basic column/row transformations. 

Some basic properties of Pfaffian are given for later use. One is the bilinear identities, as an analogy of the determinant identities such as Jacobi identity and Pl\"ucker relation. These identities play important roles in many different contexts of soliton theory \cite{hirota04,ohta92,ohta04} and algebraic combinatorics  \cite{eisenkolbl13,knuth95,stembridge90,yan07}.
Although some of the low order Pfaffian identities was shown in \cite{tanner78},  the most general cases were given in Prof Yasuhiro Ohta's PhD thesis \cite{ohta92}. The bilinear identities are written as
\begin{subequations}
\begin{align}
\Pf(a_1,\cdots,a_{m},\star)\Pf(\star)&=\sum_{j=2}^{m}(-1)^j \Pf(a_1,a_j,\star)\Pf(a_2,\cdots,\hat{a}_j,\cdots,a_{m},\star),\label{even}\\
\Pf(a_1,\cdots,a_m,\ast)\Pf(\ast,2n)&=\sum_{j=1}^{m}(-1)^{j-1}\Pf(a_j,\ast)\Pf(a_1,\cdots,\hat{a}_j,\cdots,a_m,\ast,2n),\label{odd}
\end{align}
\end{subequations}
where $\{\star\}=\{1,\cdots,2n\}$ and $\{\ast\}=\{1,\cdots,2n-1\}$.

Since the discrete integrable system and discrete Pfaffian tau functions are the main objects considered, some basic formulae for the discrete Pfaffian elements should be introduced. One is the discrete Gram-type Pfaffian. If $\Pf(i^\ast,j^\ast)=\Pf(i,j)+\lambda \Pf(a,b,i,j)$ ($\lambda\in\mathbb{C}\backslash\{0\}$) and $\Pf(a,b)=0$, then
\begin{align}\label{gram1}
\Pf(1^\ast,\cdots,2n^\ast)=\Pf(1,\cdots,2n)+\lambda\Pf(a,b,1,\cdots,2n).
\end{align}
Moreover, if we further have $\Pf(a,i^\ast)=\Pf(a,i)+\lambda\Pf(b,i)$, then
\begin{align}\label{gram2}
\Pf(a,1^\ast,\cdots,2n-1^\ast)=\Pf(a,1,\cdots,2n-1)+\lambda\Pf(b,1,\cdots,2n-1).
\end{align}
Another discrete Pfaffian is called the Wronski-type, which is also called as the addition formula for Pfaffians  \cite{hirota04,ohta04}.
The elements in this case satisfy $$\Pf(i^\ast,j^\ast)=\lambda^2\Pf(i,j)+\lambda\Pf(i+1,j)+\lambda\Pf(i,j+1)+\Pf(i+1,j+1),$$ and the addition formula gives
\begin{align}\label{wronski1}
\Pf(1^\ast,\cdots,2n^\ast)=\Pf(c,1,\cdots,2n+1), \quad \text{where $\Pf(c,i)=(-\lambda)^{i-1}$.}
\end{align}
Furthermore, if we have $\Pf(d,i^\ast)=\lambda\Pf(d,i)+\Pf(d,i+1)$, then
\begin{align}\label{wronski2}
\Pf(d,1^\ast,\cdots,2n-1^\ast)=\Pf(d,c,1,\cdots,2n),\quad \text{where $\Pf(d,c)=0$}.
\end{align}
Please refer  to \cite{hirota13,ohta04} for more details about the Wronski-type discrete Pfaffians.

\section{Condensation algorithms for Pfaffian and integrable lattices}\label{ca}
In this section, we consider two integrable lattices which could be regarded as condensation algorithms for Pfaffians. The first one comes from the famous Miwa equation or so-called discrete BKP (dBKP) equation, which has been extensively studied in soliton theory. Its discrete soliton solutions were exhibited in \cite{tsujimoto96} and recently the molecule solutions, i.e. the solutions with special initial values, were given in \cite{chang18}. Interestingly, this lattice equation was also proposed in combinatorics and named as the cube recurrence \cite{propp01}. Its Laurent property was shown by Fomin and Zelevinsky \cite{fomin02} and corresponding combinatorial objects called groves were given in \cite{carroll04}. We consider its B\"acklund transformation as a condensation algorithm for Pfaffians. The other integrable lattice algorithm we consider here is the Toda lattice of discrete DKP type, which is a coupled lattice equation. This lattice equation was proposed when we studied the Pfaffian sequence transformations \cite{chang182} and the integrability conditions such as B\"acklund transformation and Lax pair were shown as well.

 We can see that these two lattice equations are totally different but they are all efficient to compute the values of Pfaffians.

\subsection{A condensation algorithm for Pfaffian---an integrable lattice of $B$-type}
Let's start with the famous Miwa equation 
\begin{align}\label{btoda}
\tau_n^{k+1,l+1}\tau_n^{k,l}-\tau_n^{k,l+1}\tau_n^{k+1,l}=\tau_{n-1}^{k+1,l+1}\tau_{n+1}^{k,l}-\tau_{n+1}^{k+1,l}\tau_{n-1}^{k,l+1},
\end{align}
where we take the discrete step length as $1$.

As demonstrated in \cite{chang18}, with initial values $\tau_{-1}^{k,l}=0$, $\tau_0^{k,l}=1$, $\tau_1^{k,l}=\Pf(d_0,0)^{k,l}$ and $\tau_2^{k,l}=\Pf(0,1)^{k,l}$, where $\Pf(d_0,0)^{k,l}$ and $\Pf(0,1)^{k,l}$ are some functions given, this lattice equation admits the following Pfaffian tau-functions
\begin{align*}
\tau_{2n}^{k,l}=\Pf(0,\cdots,2n-1)^{k,l},\quad \tau^{k,l}_{2n+1}=\Pf(d_0,0,\cdots,2n)^{k,l},
\end{align*}
where the Pfaffian elements satisfy the discrete evolutions
\begin{align*}
\Pf(d_0,i)^{k+1,l}&=\Pf(d_1,i)^{k,l},\\
\Pf(i,j)^{k+1,l}&=\Pf(i,j)^{k,l}+\Pf(d_0,d_1,i,j)^{k,l},\\
\Pf(d_0,i)^{k,l+1}&=\Pf(d_0,i)^{k,l}+\Pf(d_0,i+1)^{k,l},\\
\Pf(i,j)^{k,l+1}&=\Pf(i,j)^{k,l}+\Pf(i+1,j)^{k,l}+\Pf(i,j+1)^{k,l}+\Pf(i+1,j+1)^{k,l}.
\end{align*}
It can be verified directly by using the bilinear identities \eqref{even} and \eqref{odd} with the help of discrete Pfaffian elements introduced in \eqref{gram1}-\eqref{wronski2}. Unfortunately, although we can give them explicit solutions, this scheme is not applicable to iterate, so it is not helpful in the construction of algorithm and we leave it here without proof; please refer to \cite{chang18} for some hints of verification. In other words, if we get the values of $\tau_{n-1}^{k,l}$ and $\tau_n^{k,l}$ for all $k,\,l\in\mathbb{Z}$, we need to solve a first-order difference equation to obtain the values $\tau_{n+1}^{k,l}$ (See Fig. \ref{p1}), which are not easy to ensure the uniqueness of the solutions even though all of the initial values are well set.

\begin{figure}[htbp]
\begin{center}
\begin{tikzpicture}
\newcommand*\kante{2}
\path[scale=\kante]
  (0,0,0)  coordinate   (A)
  (1,0,0)  coordinate   (B)
  (1,1,0)  coordinate   (C) 
  (0,1,0)  coordinate   (D)
  (0,0,-1) coordinate   (E)
  (1,0,-1) coordinate  (F)
  (1,1,-1) coordinate  (G)
  (0,1,-1) coordinate   (H)
   (0,2,0)  coordinate (A1)
  (1,2,0)  coordinate  (B1)
  (0,2,-1) coordinate   (E1)
  (1,2,-1) coordinate  (F1)

  (3.5,1,0)     coordinate (O)
  (3.5,1.5,0)  coordinate [label=above:$n$] (N)
  (3.5,1,-.75) coordinate [label=right:$l$](L)
  (3,1,0)  coordinate [label=below:$k$] (K)
  
;
\draw[font=\tiny]
(A) -- (A1) -- (B1) -- (B) -- cycle

(B) -- (F) -- (F1) -- (B1)

(D) -- (C) -- (G)

(A1) -- (B1) -- (F1)--(E1)--cycle
;

\draw[dashed, thick]
(A1)--(H)--(G)--(B1)--cycle

;

\path[dashed, very thin] (E) edge (A) edge (F) edge (H);
\path[dashed, very thin] (H) edge (D) edge (G) edge (E1);

\foreach \n in {D,H,C,G,E,F}
  \node at (\n)[circle,fill=gray,inner sep=1.5pt]{};
  
\foreach \l in{A1,B1}
  \node at (\l)[circle,fill,inner sep=1.5pt]{};

\draw [->,>=stealth] (O) -- (N) ;
\draw [->,>=stealth] (O) -- (L);
\draw [->,>=stealth] (O) -- (K);

\end{tikzpicture}\caption{The graphic explanation for $B$-Toda lattice: One can hardly compute the explicit value for the points at the top.}\label{p1}
\end{center}
\end{figure}
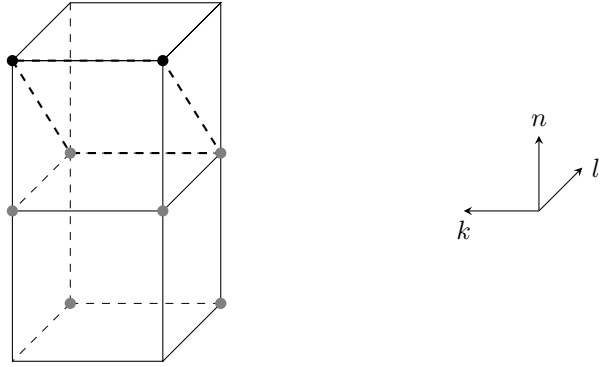

One needs to consider whether there is any explicitly iterative algorithm to compute the value of a Pfaffian. A possible way is to rotate the thickly dashed parallelogram (so as the corresponding hexahedron), and make one of the vertices at the top, two at the middle, and the rest one at the bottom. In this case, we can explicitly compute the value at the top if the values of the lower levels are known. In fact, this kind of realisation is equivalent to consider the  B\"acklund transformation of the lattice equation. As is known, the B\"acklund transformation of the discrete BKP equation was given by Gilson et al \cite{gilson03}, and later it was called as the generalised Lotka-Volterra equation \cite{chang18} 
\begin{align}\label{glv}
\tau_{n+2}^{k,l}\tau_{n-1}^{k+1,l+1}-\tau_n^{k,l+1}\tau_{n+1}^{k+1,l}=\tau_n^{k+1,l+1}\tau_{n+1}^{k,l}-\tau_n^{k+1,l}\tau_{n+1}^{k,l+1}.
\end{align}
The reasons we call it as the generalised Lotka-Volterra lattice are: (1) Just as the Lotka-Volterra lattice is the B\"acklund transformation of the Toda lattice, the generalised Lotka-Volterra lattice is the B\"acklund transformation of the B-Toda lattice; (2) If we dismiss the discrete variables $k$ or $l$, it looks similar to the full discrete Lotka-Volterra equation. The solution of this lattice is given in the following proposition.
\begin{proposition}\label{pp1}
The solutions of the generalised Lotka-Volterra equation \eqref{glv} are given by
\begin{align*}
\tau_{2n}^{k,l}=\Pf(0,\cdots,2n-1)^{k,l},\quad \tau_{2n+1}^{k,l}=\Pf(d,0,\cdots,2n)^{k,l}
\end{align*}
with the discrete Pfaffian evolutions satisfying
\begin{align}\label{de}
\begin{aligned}
\Pf(d,i)^{k+1,l}&=\Pf(d,i+1)^{k,l},\\
\Pf(i,j)^{k+1,l}&=\Pf(i+1,j+1)^{k,l},\\
\Pf(d,i)^{k,l+1}&=\Pf(d,i)^{k,l}+\Pf(d,i+1)^{k,l}.\\
\Pf(i,j)^{k,l+1}&=\Pf(i,j)^{k,l}+\Pf(i+1,j)^{k,l}+\Pf(i,j+1)^{k,l}+\Pf(i+1,j+1)^{k,l}.
\end{aligned}
\end{align}
Moreover, the initial values of the equation are given by $\tau_{-1}^{k,l}=0$, $\tau_0^{k,l}=1$, $\tau_1^{k,l}=\Pf(d,i)^{k,l}$ and $\tau_2^{k,l}=\Pf(0,1)^{k,l}$ for all $k,\,l\in\mathbb{Z}_{\geq0}$. 
\end{proposition}
\begin{proof}
The proof of this proposition is based on the discrete Wronskian-type Pfaffian and the bilinear identities. Although some hints have been demonstrated in \cite{chang18}, for self-consistency, we give more details about the proof and the ideas would be used later.
According to the Wronski-type Pfaffian formulae \eqref{wronski1} and \eqref{wronski2}, we know another label $c$ should be introduced such that $\Pf(d,c)=0$ and $\Pf(c,i)=(-1)^i$. Then from the discrete evolutions \eqref{de}, one has
\begin{align}\label{tau1}
\begin{aligned}
&\tau_{2n}^{k,l+1}=\Pf(c,0,\cdots,2n)^{k,l},\quad\qquad\quad\,\, \tau_{2n+1}^{k,l+1}=\Pf(d,c,0,\cdots,2n+1)^{k,l},\\
&\tau_{2n}^{k+1,l}=\Pf(1,\cdots,2n)^{k,l},\quad\qquad\qquad\,\, \tau_{2n+1}^{k+1,l}=\Pf(d,1,\cdots,2n+1)^{k,l},\\
&\tau_{2n}^{k+1,l+1}=-\Pf(c,1,\cdots,2n+1)^{k,l},\quad \tau_{2n+1}^{k+1,l+1}=-\Pf(d,c,1,\cdots,2n+2)^{k,l}.
\end{aligned}
\end{align}
Using the bilinear identity \eqref{even} by taking $\{a_1,a_2,a_3,a_4\}$ as $\{d,c,0,2n+1\}$ and $\{\star\}$ as $\{1,\cdots,2n\}$,
we can be lead to the bilinear identity
\begin{align*}
\tau_{2n+2}^{k,l}\tau_{2n-1}^{k+1,l+1}=-\tau_{2n}^{k+1,l}\tau_{2n+1}^{k,l+1}+\tau_{2n}^{k+1,l+1}\tau_{2n+1}^{k,l}+\tau_{2n}^{k,l+1}\tau_{2n+1}^{k+1,l}.
\end{align*}
Similarly, if we make the use of bilinear identity \eqref{odd} by taking $\{a_1,a_2,a_3\}$ as $\{d,c,0\}$ and $\{\ast\}$ as $\{1,\cdots,2n-1\}$, it reads
\begin{align*}
\tau_{2n+1}^{k,l}\tau_{2n-2}^{k+1,l+1}=-\tau_{2n-1}^{k+1,l}\tau_{2n}^{k,l+1}+\tau_{2n-1}^{k+1,l+1}\tau_{2n}^{k,l}+\tau_{2n-1}^{k,l+1}\tau_{2n}^{k+1,l}.
\end{align*}
Combining these gives the equation \eqref{glv}.
\end{proof}
\begin{remark}
It is remarkable that the $14$-point scheme demonstrated by King and Schief \cite{king14} can be realised by using one B-Toda lattice \eqref{btoda} (8-points scheme) and B\"acklund transformation \eqref{glv} in three different directions (in every direction $2$ points are added).
\end{remark}
\subsubsection{Why is it a condensation algorithm for the Pfaffian?} The process is divided into two parts.

(1) Preparation of Data. Introducing quadruplets $\Pf(i,j,k,l)$ to denote the Pfaffian elements $\Pf(i,j)^{k,l}$. Let's consider a $2N\times 2N$ skew symmetric matrix $\left(a_{i,j}\right)_{i,j=0}^{2N-1}$ and store the elements $a_{i,j}$ as $\Pf(i,j,0,0)$.
Following the discrete evolution relation \eqref{de}, one can iteratively obtain 
 $\Pf(i,j,k,0):=\Pf(i+k,j+k,0,0)$ for $1\leq k\leq 2N-2$ and
\begin{align*}
&\Pf(i,j,k,l):=\\
&\quad\quad\Pf(i,j,k,l-1)+\Pf(i+1,j,k,l-1)+\Pf(i,j+1,k,l-1)+\Pf(i+1,j+1,k,l-1)
\end{align*}
 for $1\leq k+l\leq 2N-2$. Therefore, information about $\Pf(0,1,k,l)$ are stored for $0\leq k+l\leq 2N-2$. Introduce triplets $\md(i,k,l)$ to take the place of $\Pf(d_0,i)^{k,l}$, which are totally independent with $\Pf(i,j,k,l)$ (c.f. \eqref{de}). We can assume that there be $2N$ free parameters $\{\alpha_i\}_{i=0}^{2N-1}$, such that $\md(i,0,0)=\alpha_i$. Following \eqref{de} again, we know that $\md(i,k,0):=\md(i+k,0,0)$ and $\md(i,k,l):=\md(i,k,l-1)+\md(i+1,k,l-1)$. Therefore, we get enough information about $\md(0,k,l)$ for $0\leq k+l\leq 2N-2$ for iteration.

(2) Iteration. Being well prepared, we now take $\phi_n^{k,l}:=\tau_{2n}^{k,l}$ and $\psi_n^{k,l}:=\tau_{2n+1}^{k,l}$ and initial values $\phi_0^{k,l}=1$, $\phi_1^{k,l}=\Pf(0,1,k,l)$ and $\psi_0^{k,l}=\md(0,k,l)$. By realising that the equation \eqref{glv} can be split into two equations
\begin{subequations}
\begin{align}
\psi_n^{k,l}\phi_{n-1}^{k+1,l+1}&=-\psi_{n-1}^{k+1,l}\phi_n^{k,l+1}+\psi_{n-1}^{k+1,l+1}\phi_n^{k,l}+\psi_{n-1}^{k,l+1}\phi_n^{k+1,l},\label{glv1}\\
\phi_{n+1}^{k,l}\psi_{n-1}^{k+1,l+1}&=-\phi_n^{k+1,l}\psi_n^{k,l+1}+\phi_n^{k+1,l+1}\psi_n^{k,l}+\phi_n^{k,l+1}\psi_n^{k+1,l},\label{glv2}
\end{align}
\end{subequations}
one can use the equation \eqref{glv1} to obtain $\psi_1^{k,l}$ and use \eqref{glv2} to obtain $\phi_2^{k,l}$, etc.. 
The value of $\phi_N^{0,0}$ is the goal for iteration, to obtain the exact value of a Pfaffian of order $N$.

The computational cost is mainly from the preparation process and the algorithm needs about $O(N^4)$ floating-point operations. We demonstrate a one-step process in Fig \ref{p2}.

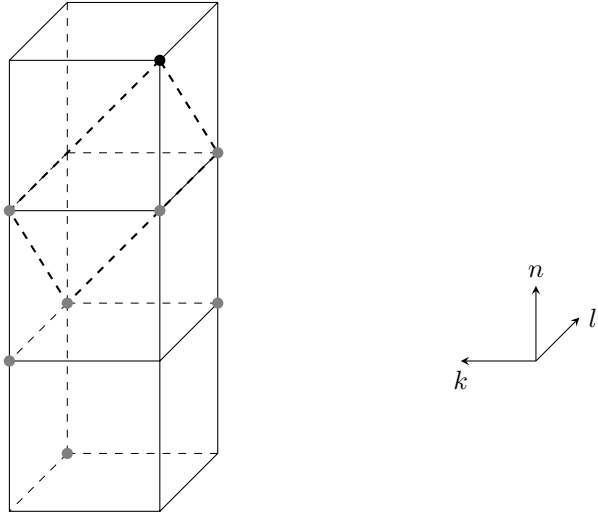
\begin{figure}[htbp]
\centering
\begin{tikzpicture}
\newcommand*\kante{2}
\path[scale=\kante]
  (0,0,0)  coordinate   (A)
  (1,0,0)  coordinate   (B)
  (1,1,0)  coordinate   (C) 
  (0,1,0)  coordinate   (D)
  (0,0,-1) coordinate   (E)
  (1,0,-1) coordinate  (F)
  (1,1,-1) coordinate  (G)
  (0,1,-1) coordinate   (H)
   (0,2,0)  coordinate (A1)
  (1,2,0)  coordinate  (B1)
  (1,3,0)  coordinate  (C1)
  (0,3,0)  coordinate   (D1)
  (0,2,-1) coordinate   (E1)
  (1,2,-1) coordinate  (F1)
  (1,3,-1) coordinate  (G1)
  (0,3,-1) coordinate  (H1)
  
  (3.5,1,0)     coordinate (O)
  (3.5,1.5,0)  coordinate [label=above:$n$] (N)
  (3.5,1,-.75) coordinate [label=right:$l$](L)
  (3,1,0)  coordinate [label=below:$k$] (K)
  
;
\draw[font=\tiny]
(A) -- (D1) -- (C1) -- (B) -- cycle

(B) -- (F) -- (G1) -- (H1) -- (D1)

(C1) -- (G1)

(D) -- (C) -- (G)

(A1) -- (B1) -- (F1)
;

\draw[dashed, thick]
(C1)--(A1)--(H)--(F1)--cycle;

\path[dashed, very thin] (E) edge (A) edge (F) edge (H);
\path[dashed, very thin] (H) edge (D) edge (G) edge (E1);
\path[dashed, very thin] (E1) edge (A1) edge (F1) edge (H1);

\foreach \n in {D,E,G,H,A1,B1,F1}
  \node at (\n)[circle,fill=gray,inner sep=1.5pt]{};

  \node at (C1)[circle,fill=black,inner sep=1.5pt]{};

\draw [->,>=stealth] (O) -- (N) ;
\draw [->,>=stealth] (O) -- (L);
\draw [->,>=stealth] (O) -- (K);

\end{tikzpicture}
\caption{The graphic explanation for Generalised Lotka-Volterra lattice: The exact value of the point at the top will be explicitly obtained if enough initial values are given.}\label{p2}
\end{figure}
{\bf{An illustrating example:}}
Consider the easiest example $\Pf(a_{i,j})_{i,j=0}^3$ with Pfaffian elements $a_{i,j}=-a_{j,i}$ and introduce the free parameters $\alpha_i=1$. Some simple computations tell us that $\psi_0^{k,l}=2^{l}$, and $\phi_1^{0,0}=a_{0,1}$, $\phi_1^{1,0}=a_{1,2}$, $\phi_1^{0,1}=a_{1,2}+a_{0,2}+a_{0,1}$, $\phi_1^{1,1}=a_{2,3}+a_{1,3}+a_{1,2}$, $\phi_1^{2,0}=a_{2,3}$ and $\phi_1^{0,2}=a_{2,3}+2a_{1,3}+3a_{1,2}+2a_{0,2}+a_{0,3}+a_{0,1}$, which are the preparations for the iterations. Now we can compute $\psi_1^{0,0}=-\phi_1^{0,1}+2\phi_1^{0,0}+2\phi_1^{1,0}=a_{0,1}-a_{0,2}+a_{1,2}$, $\psi_1^{0,1}=-2\phi_1^{0,2}+4\phi_1^{0,1}+4\phi_1^{1,1}=2(a_{2,3}+a_{1,2}-a_{0,3}+a_{0,1})$ and $\psi_1^{1,0}=-\phi_1^{1,1}+2\phi_1^{1,0}+2\phi_1^{2,0}=a_{1,2}-a_{1,3}+a_{2,3}$. Putting them into \eqref{glv2}, we can get
$\phi_2^{0,0}=a_{0,1}a_{2,3}-a_{0,2}a_{1,3}+a_{0,3}a_{1,2}$, which coincides with the expansion of the second order Pfaffian.

In fact, regarding the lattice equation \eqref{glv} and discrete evolutions \eqref{de}, we simply took the $\lambda=1$ in the discrete Wronski formulae (c.f. eqs. \eqref{wronski1}-\eqref{wronski2}). However, we can keep the parameter $\lambda$ inside the Pfaffian elements, and then the system can be generalised to the one involving a relaxation parameter. Consider $\lambda\in\mathbb{C}\backslash\{0\}$ and the discrete evolutions
\begin{align*}
&\Pf(d,i)^{k+1,l}=\Pf(d,i+1)^{k,l},\\
&\Pf(i,j)^{k+1,l}=\Pf(i+1,j+1)^{k,l},\\
&\Pf(d,i)^{k,l+1}=\lambda\Pf(d,i)^{k,l}+\Pf(d,i+1)^{k,l},\\
&\Pf(i,j)^{k,l+1}=\lambda^2\Pf(i,j)^{k,l}+\lambda\Pf(i+1,j)^{k,l}+\lambda\Pf(i,j+1)^{k,l}+\Pf(i+1,j+1),
\end{align*}
then the discrete evolutions of the tau functions should be the same with \eqref{tau1} except
\begin{align*}
\tau_{2n}^{k+1,l+1}=-\lambda^{-1}\Pf(c,1,\cdots,2n+1)^{k,l},\quad
\tau_{2n+1}^{k+1,l+1}=-\lambda^{-1}\Pf(d,c,1,\cdots,2n+2)^{k,l}.
\end{align*}
Following the proof in Proposition \ref{pp1} and using bilinear identities backwards, one can find the following iteration scheme
\begin{align}\label{relax1}
\begin{aligned}
&\tau_{2n+2}^{k,l}\tau_{2n-1}^{k+1,l+1}=\tau_{2n}^{k+1,l+1}\tau_{2n+1}^{k,l}+\lambda \tau_{2n}^{k,l+1}\tau_{2n+1}^{k+1,l}-\lambda \tau_{2n}^{k+1,l}\tau_{2n+1}^{k,l+1},\\
&\tau_{2n+1}^{k,l}\tau_{2n-2}^{k+1,l+1}=\tau_{2n-1}^{k+1,l+1}\tau_{2n}^{k,l}+\lambda\tau_{2n-1}^{k,l+1}\tau_{2n}^{k+1,l}-\lambda\tau_{2n-1}^{k+1,l}\tau_{2n}^{k,l+1}.\end{aligned}
\end{align}
We call the parameter $\lambda$ as the relaxation factor and when $\lambda=1$, it reduces to the original algorithm.

\subsection{Another condensation algorithm for Pfaffian---an integrable lattice of $D$-type}
This part is devoted to another condensation algorithm for the Pfaffian. This integrable lattice was firstly proposed by considering a reasonable acceleration for the Pfaffian sequence transformation with its integrability \cite{chang182}. Unlike the generalised Lotka-Volterra lattice, which could be written as a unified bilinear recurrence relation, the integrable lattice related to the DKP equation can only be written in a coupled formalism. The main results are stated as following proposition.
\begin{proposition}\label{pp2}
The Toda lattice of DKP type admits the form
\begin{subequations}
\begin{align}
&\sigma_n^{k,l+1}\tau_{n-1}^{k+1,l+1}=\tau_{n-1}^{k+1,l+1}\tau_n^{k,l+1}+\tau_{n-1}^{k+1,l+2}\tau_n^{k,l}-\tau_{n-1}^{k,l+2}\tau_n^{k+1,l}+\sigma_{n-1}^{k+1,l+1}\tau_n^{k,l+1},\label{deq1}\\
&\tau_{n+1}^{k,l}\tau_{n-1}^{k+1,l+2}=\sigma_n^{k,l+1}\tau_n^{k+1,l+1}-\sigma_n^{k+1,l+1}\tau_n^{k,l+1}-\tau_n^{k+1,l+1}\tau_n^{k,l+1}+\tau_n^{k,l+2}\tau_n^{k+1,l}.\label{deq2}
\end{align}
\end{subequations}
By considering the proper initial values $
\tau_0^{k,l}=1,\,\tau_1^{k,l}=\Pf(0,1)^{k,l},\, \sigma_0^{k,l}=0
$ where $\Pf(0,1)^{k,l}$ are some given values,
the solutions of this integrable lattice are given in terms of Pfaffian, namely
\begin{align*}
\tau_n^{k,l}=\Pf(0,\cdots,2n-1)^{k,l},\quad \sigma_n^{k,l}=\Pf(d_0,d_1,0,\cdots,2n-1)^{k,l}
\end{align*}
with the Pfaffian elements satisfying
\begin{subequations}
\begin{align}
&\Pf(d_0,d_1)^{k,l}=0,\quad \Pf(d_0,i)^{k,l}=1,\label{sde0}\\
&\Pf(d_1,i)^{k,l+1}=\Pf(0,i+1)^{k,l}-\Pf(0,i)^{k,l},\label{sde4}\\
&\Pf(i,j)^{k,l+1}=\Pf(i,j)^{k,l}-\Pf(i+1,j)^{k,l}-\Pf(i,j+1)^{k,l}+\Pf(i+1,j+1)^{k,l},\label{sde1}\\
&\Pf(d_1,i)^{k+1,l}=\Pf(d_1,i+1)^{k,l}+\Pf(0,i+1)^{k,l},\label{sde2}\\
&\Pf(i,j)^{k+1,l}=\Pf(i+1,j+1)^{k,l}.\label{sde3}
\end{align}
\end{subequations}
\end{proposition}

\begin{proof}
The key point is to make the use of \eqref{trans}. By abstracting the second line by the first one and continuing the process, one can find
\begin{align*}
\Pf(0,\cdots,2n-1)=\Pf(d_1,{0}^\ast,\cdots,{2n-2}^\ast),
\end{align*}
where $\Pf(d_1,{i}^\ast)=\Pf(0,i+1)-\Pf(0,i)$ and $\Pf({i}^\ast,{j}^\ast)=\Delta_1\Delta_2\Pf(i,j)$, with notations $\Delta_1\Pf(i,j)=\Pf(i+1,j)-\Pf(i,j)$, $\Delta_2\Pf(i,j)=\Pf(i,j+1)-\Pf(i,j)$.
Such observations lead to the discrete evolutions \eqref{sde4}-\eqref{sde1}. Therefore, by discrete evolutions \eqref{sde0}-\eqref{sde3}, one can find
\begin{align*}
&\Pf(d_0,0,\cdots,2n)^{k,l}=\tau_n^{k,l+1},\quad \Pf(1,\cdots,2n)^{k,l}=\tau_n^{k+1,l},\quad \Pf(d_1,0,\cdots,2n)^{k,l}=\tau_{n+1}^{k,l-1},\\
&\Pf(d_1,1,\cdots,2n+1)^{k,l}=\tau_{n+1}^{k+1,l-1}-\tau_{n+1}^{k,l},\quad
\Pf(d_0,1,\cdots,2n+1)^{k,l}=\tau_{n}^{k+1,l+1},\\
&\Pf(d_0,d_1,1,\cdots,2n)^{k,l}=\sigma_n^{k+1,l}-\tau_n^{k,l+1}+\tau_{n}^{k+1,l}.
\end{align*}
Making use of the Pfaffian identity \eqref{even} with $\{a_1,a_2,a_3,a_4\}=\{d_0,d_1,0,2n+1\}$ and $\{\star\}=\{1,\cdots,2n\}$ and identity \eqref{odd} with $\{a_1,a_2,a_3\}=\{d_0,d_1,0\}$ and $\{\ast\}=\{1,\cdots,2n-1\}$, one can verify that the Pfaffian tau functions with the given discrete evolutions satisfy the integrable lattice.
\end{proof}

It's not surprising that \eqref{deq1}-\eqref{deq2} can be used for computing the value of Pfaffian since it has been utilised as an iterative algorithm in convergent acceleration algorithm.
As before, we would show how to realise it as a condensation algorithm.
Let's consider a $2N\times 2N$ skew symmetric matrix $\left(a_{i,j}\right)_{i,j=0}^{2N-1}$, and store them in the quadruplet $\pf(i,j,0,0)=a_{i,j}$. Then we can set $\pf(i,j,k,0):=\pf(i+k,j+k,0,0)$ for $1\leq k\leq 2N-1$ and 
\begin{align*}
&\pf(i,j,k,l):=\\&
\quad\quad
\pf(i,j,k,l-1)-\pf(i,j+1,k,l-1)-\pf(i+1,j,k,l-1)+\pf(i+1,j+1,k,l-1)
\end{align*} for $1\leq k+l\leq 2N-1$. 
Therefore, by setting the initial values
\begin{align*}
\tau_0^{k,l}=1,\,\tau_1^{k,l}=\pf(0,1,k,l),\, \sigma_0^{k,l}=0,
\end{align*}
the equation \eqref{deq1} would give us the exact value of $\sigma_1^{k,l}$ and \eqref{deq2} would give us the value of $\tau_2^{k,l}$, etc.. Iterating the values of $\sigma$ and $\tau$ repeatedly, we would finally obtain the results of $\tau_n^{0,0}$, as expected. The computational cost is mainly from the storage the values for $\pf(0,1,k,l)$, which cost $O(N^4)$ floating-point operations.

{\bf{An illustrating example:}}
Let's consider a second order Pfaffian $\Pf(a_{i,j})_{i,j=0}^3$ with $a_{i,j}=-a_{j,i}$, then we know $\tau_0^{0,0}=a_{0,1}$, $\tau_0^{0,1}=a_{1,2}-a_{0,2}+a_{0,1}$, $\tau_0^{1,0}=a_{1,2}$, $\tau_0^{1,1}=a_{2,3}-a_{1,3}+a_{1,2}$, $\tau_0^{2,0}=a_{2,3}$ and $\tau_0^{0,2}=a_{2,3}-2a_{1,3}+3a_{1,2}+a_{0,3}-2a_{0,2}+a_{0,1}$.
Therefore, from equation \eqref{deq1} we can compute $\sigma_1^{0,1}=2a_{0,1}-a_{0,2}$
and $\sigma_{1}^{1,1}=2a_{1,2}-a_{1,3}$, and then from equation \eqref{deq2}, we know that
$\tau_2^{0,0}=a_{0,1}a_{2,3}-a_{0,2}a_{1,3}+a_{0,3}a_{1,2}$.

Similarly, we can introduce a free parameter into the relations \eqref{sde4}-\eqref{sde1} and consider the following discrete evolutions
\begin{align*}
&\Pf(d_1,i)^{k,l+1}=\Pf(0,i+1)^{k,l}+\lambda\Pf(0,i)^{k,l},\\
&\Pf(i,j)^{k,l+1}=\lambda^2\Pf(i,j)^{k,l}+\lambda\Pf(i+1,j)^{k,l}+\lambda\Pf(i,j+1)^{k,l}+\Pf(i+1,j+1)^{k,l}.
\end{align*}
According to the proof of Proposition \ref{pp2}, this kind of evolution comes from the basic column/row transformations by adding the second line by the first line times $\lambda$. By simply expanding the Pfaffian and making the use of discrete evolutions, one can find there are only two terms changed
\begin{align*}
\Pf(d_0,1,\cdots,2n+1)^{k,l}=-\lambda\tau_n^{k+1,l+1},\quad \Pf(d_0,d_1,1,\cdots,2n)^{k,l}=-\lambda\sigma_n^{k+1,l}-\tau_n^{k,l+1}+\tau_n^{k+1,l}.
\end{align*}
Following Proposition \ref{pp2}, one can find a relaxation factor appeared in the algorithm 
\begin{align}\label{relax2}
&\sigma_{n}^{k,l+1}\tau_{n-1}^{k+1,l+1}=\tau_{n-1}^{k+1,l+1}\tau_n^{k,l+1}-\lambda\tau_{n-1}^{k+1,l+2}\tau_n^{k,l}-\tau_{n-1}^{k,l+2}\tau_n^{k+1,l}-\lambda\sigma_{n-1}^{k+1,l+1}\tau_n^{k,l+1},\\
&\tau_{n+1}^{k,l}\tau_{n-1}^{k+1,l+2}=-\lambda^{-1}\sigma_n^{k,l+1}\tau_{n}^{k+1,l+1}-\sigma_n^{k+1,l+1}\tau_n^{k,l+1}+\lambda^{-1}\tau_n^{k,l+1}\tau_n^{k+1,l+1}-\lambda^{-1}\tau_n^{k+1,l}\tau_n^{k,l+2},\nonumber
\end{align}
which reduces to \eqref{deq1}-\eqref{deq2} when $\lambda=-1$.
\section{Discretisation of Gram-type Pfaffian and discrete integrable system}\label{dc}
In this part, we demonstrate another discretisation of integrable systems by using the discrete Pfaffian elements, to show the effectiveness of the method.
In fact, the solution we consider in this part is a kind of determinant, but from the connection between special Pfaffian and 
determinant, we know that it can be dealt with the formulae by Pfaffian. The link between determinant and Pfaffian was elaborately discussed in, for example, \cite[Section 2]{ohta04}.

 The lattice equation considered here is the C-Toda lattice \cite{chang183,li19}, which has a $2+1$-dimensional generalisation
\begin{align}\label{2+1}
D_t\tau_{n+1}\cdot\tau_n=\sigma_n^2,\quad D_tD_s\tau_{n+1}\cdot\tau_{n+1}=4\sigma_{n+1}\sigma_n.
\end{align}
It is remarkable that the equation could be iterated if the initial values are given; one can get the exact values of $\tau_{n+1}$ from the first equation and that of $\sigma_{n+1}$ from the second equation although a first-order ODE should be solved. Moreover, 
the solutions of this lattice equation have closed determinant forms \cite{chang183} 
\begin{align*}
\tau_n=\det(I_{i,j})_{i,j=0,\cdots,n-1},\quad \sigma_n=\det\left(I_{i,j},\,\alpha_j\right)_{\substack{i=0,\cdots,n\\j=0,\cdots,n-1}}
\end{align*}
with the time evolutions
\begin{align}\label{2+1t}
\frac{\p}{\p t}I_{i,j}=I_{i+1,j}+I_{i,j+1},\quad \frac{\p}{\p s}I_{i,j}=\alpha_i\alpha_j,\quad \frac{\p}{\p t}\alpha_i=\alpha_{i+1}.
\end{align}
The Pfaffian version of the solutions was given in \cite{li19} by recognising that
\begin{align*}
\tau_n=\Pf(0,\cdots,n-1,n-1^*,\cdots,0^*),\quad \sigma_n=(-1)^n\Pf(d_0,0,\cdots,n-1,n^*,\cdots,0^*)
\end{align*}
with Pfaffian entries $\Pf(i,j)=\Pf(i^*,j^*)=\Pf(d_0,i)=0$, $\Pf(i,j^*)=I_{i,j}$ and $\Pf(d_0,i^*)=\alpha_i$. With these notations, the time evolution relations \eqref{2+1t} could be rewritten in terms of Pfaffian as
\begin{subequations}
\begin{align}
&\p_t\Pf(i,j^*)=\Pf(i+1,j^*)+\pf(i,j+1^*),\quad \p_t\Pf(d_0,i^*)=\Pf(d_0,i+1^*),\label{timet}\\
&\p_s\Pf(i,j^*)=\Pf(d_0,d_0^*,i,j^*)\label{times}
\end{align}\end{subequations}
with $\pf(d_0^*,i)=\alpha_i$ according to the symmetric property.
Then we discretise the time evolutions of the Pfaffian elements. To demonstrate it clearly, we denote the notations by the correspondence $(s,t)\mapsto(m,l)$. Using the discretisation technique, we could construct the discrete evolutions in direction $l$ as
\begin{align*}
&\Pf(i,j^*)^{m,l+1}=\Pf(i,j^*)^{m,l}-\Pf(i+1,j^*)^{m,l}-\Pf(i,j+1^*)^{m,l}+\Pf(i+1,j+1^*)^{m,l},\\
&\Pf(d_0,i^*)^{m,l+1}=\Pf(d_0,i+1^*)^{m,l}-\Pf(d_0,i^*)^{m,l},
\end{align*}
which could be regarded as the discretisation of relation \eqref{timet}. Regarding relation \eqref{times}, one could discretise it by
\begin{align*}
\Pf(i,j^*)^{m+1,l}=\Pf(d_0,d_0^*,i,j^*)^{m,l}-\Pf(i,j^*)^{m,l}.
\end{align*}
Therefore, if we define the discrete Pfaffian $\tau$-function as
\begin{align*}
\tau_N^{m,l}=\Pf(0,\cdots,N-1,N-1^*,\cdots,0^*)^{m,l},\quad \sigma_N^{m,l}=\Pf(d_0,0,\cdots,N-1,N^*,\cdots,0^*)^{m,l},
\end{align*}
and due to the addition formula for discrete Pfaffians \eqref{wronski1}-\eqref{wronski2}, one could further obtain the following relations
\begin{align*}
&\Pf(d_0,d_0^*,0,\cdots,N-1,N-1^*,\cdots,0^*)^{m,l}=\tau_N^{m+1,l}-\tau_N^{m,l},\\
&\Pf(c_0,c_0^*,0,\cdots,N-1,N-1^*,\cdots,0^*)^{m,l}=\tau_{N-1}^{m,l+1},\\
&\Pf(d_0,c_0,c_0^*,0,\cdots,N-1,N^*,\cdots,0^*)^{m,l}=\sigma_{N-1}^{m,l+1},
\end{align*}
where Pfaffian entries \begin{align*}
\text{$\Pf(c_0,i^*)^{m,l}=\Pf(c_0^*,i)^{m,l}=1$, $\Pf(c_0,c_0^*)^{m,l}=\Pf(c_0,i)^{m,l}=\Pf(c_0^*,i^*)^{m,l}=0$.}
\end{align*}
By using the Pfaffian identity \eqref{even}
with $\{\star\}=\{0,\cdots,N-1,N-1^*,\cdots,0^*\}$, one could get
\begin{align}\label{c1}
\tau_{N+1}^{m+1,l}\tau_N^{m,l}-\tau_N^{m+1,l}\tau_{N+1}^{m,l}=(\sigma_N^{m,l})^2
\end{align}
if $\{a_0,a_1,a_2,a_3\}=\{d_0,d_0^*,N,N^*\}$ is chosen. 
Moreover, taking $\{a_0,a_1,a_2,a_3\}=\{d_0,d_0^*,c_0,c_0^*\}$ and denoting a new variable
\begin{align*}
\xi_N^{m,l}=\Pf(c_0,d_0^*,0,\cdots,N,N^*,\cdots,0^*),
\end{align*}
one gets another equation
\begin{align}\label{c2}
\tau_{N}^{m,l}\tau_{N-1}^{m+1,l+1}-\tau_{N}^{m+1,l}\tau_{N-1}^{m,l+1}=(\xi_{N-1}^{m,l})^2.
\end{align}
It should be noted that $\xi_N^{m,l}$ is a new variable which seems not appear in the continuous equation \eqref{2+1}. To make it closed, it is obvious that we need more relations between the auxiliary variables $\xi$, $\sigma$ and the tau function $\tau$. 
Therefore, a higher order Pfaffian identity is needed at this stage. By using the identity \eqref{even}, and taking $\{\star\}$ as $\{0,\cdots,N-1,N-1^*,\cdots,0^*\}$ and $\{a_0,\cdots,a_5\}$ as $\{d_0,d_0^*,c_0,c_0^*,N,N^*\}$, the identity demonstrates the following relation
\begin{align}\label{c3}
\tau_{N}^{m+1,l+1}\tau_N^{m,l}-\tau_N^{m+1,l}\tau_N^{m,l+1}=\sigma_N^{m,l}\sigma_{N-1}^{m,l+1}-\xi_N^{m,l}\xi_{N-1}^{m,l}.
\end{align}
With the help of equations \eqref{c1}, \eqref{c2} and \eqref{c3}, one can derive the following equation
\begin{align}\label{c4}
\tau_{N+1}^{m+1,l}\tau_{N-1}^{m,l+1}-\tau_{N+1}^{m,l}\tau_{N-1}^{m+1,l+1}=\sigma_N^{m,l}\sigma_{N-1}^{m,l+1}+\xi_N^{m,l}\xi_{N-1}^{m,l}.
\end{align}
Therefore, by eliminating the term $\sigma_{N}^{m,l}\sigma_{N-1}^{m,l+1}$ from equations \eqref{c3} and \eqref{c4}, it results in the equation
\begin{align}\label{c5}
\tau_{N+1}^{m+1,l}\tau_{N-1}^{m,l+1}-\tau_{N+1}^{m,l}\tau_{N-1}^{m+1,l+1}-\tau_N^{m+1,l+1}\tau_N^{m,l}+\tau_N^{m+1,l}\tau_N^{m,l+1}=2\xi_N^{m,l}\xi_{N-1}^{m,l}.
\end{align}
It is remarkable equations \eqref{c2} and \eqref{c5} could be regarded as a discretisation of $2+1$ dimensional C-Toda lattice \eqref{2+1}. Furthermore, these two equations could be written in terms of $\tau$ only, and it results in
\begin{align*}
&4(\tau_{N+1}^{m,l}\tau_N^{m+1,l+1}-\tau_{N+1}^{m+1,l}\tau_N^{m,l+1})(\tau_N^{m,l}\tau_{N-1}^{m+1,l+1}-\tau_N^{m+1,l}\tau_{N-1}^{m,l+1})\\&=(\tau_N^{m,l}\tau_N^{m+1,l+1}+\tau_{N+1}^{m,l}\tau_{N-1}^{m+1,l+1}-\tau_{N-1}^{m,l+1}\tau_{N+1}^{m+1,l}-\tau_N^{m+1,l}\tau_N^{m,l+1})^2.
\end{align*}
This equation is equivalent to the discrete CKP equation \cite{bobenko15,schief03}
\begin{align*}
&4(\hat{\tau}_{N+1}^{m+1,l}\hat{\tau}_{N+1}^{m,l+1}-\hat{\tau}_{N+1}^{m,l}\hat{\tau}_{N+1}^{m+1,l+1})(\hat{\tau}_{N}^{m+1,l}\hat{\tau}_{N}^{m,l+1}-\hat{\tau}_{N}^{m,l}\hat{\tau}_{N}^{m+1,l+1})\\
&=(\hat{\tau}_{N}^{m,l}\hat{\tau}_{N+1}^{m+1,l+1}+\hat{\tau}_{N+1}^{m,l}\hat{\tau}_{N}^{m+1,l+1}-\hat{\tau}_{N}^{m,l+1}\hat{\tau}_{N+1}^{m+1,l}-\hat{\tau}_{N}^{m+1,l}\hat{\tau}_{N+1}^{m,l+1})^2
\end{align*}
via the affine transformation $\hat{\tau}_{N}^{m,l}\mapsto \tau_{N-l}^{m,l}$. To conclude, we state the following proposition.
\begin{proposition}
The discrete C-Toda (dCKP) lattice 
\begin{align}
\begin{aligned}\label{2+1d}
&4(\tau_{N+1}^{m,l}\tau_N^{m+1,l+1}-\tau_{N+1}^{m+1,l}\tau_N^{m,l+1})(\tau_N^{m,l}\tau_{N-1}^{m+1,l+1}-\tau_N^{m+1,l}\tau_{N-1}^{m,l+1})\\&=(\tau_N^{m,l}\tau_N^{m+1,l+1}+\tau_{N+1}^{m,l}\tau_{N-1}^{m+1,l+1}-\tau_{N-1}^{m,l+1}\tau_{N+1}^{m+1,l}-\tau_N^{m+1,l}\tau_N^{m,l+1})^2.
\end{aligned}
\end{align}
admits the following solution
\begin{align*}
\tau_N^{m,l}=\Pf(0,\cdots,N-1,N-1^*,\cdots,0^*)^{m,l},
\end{align*}
where the Pfaffian elements satisfy
\begin{align*}
&\Pf(i,j^*)^{m,l+1}=\Pf(i,j^*)^{m,l}-\Pf(i+1,j^*)^{m,l}-\Pf(i,j+1^*)^{m,l}+\Pf(i+1,j+1^*)^{m,l},\\
&\Pf(i,j^*)^{m+1,l}=\Pf(d_0,d_0^*,i,j^*)^{m,l}-\Pf(i,j^*)^{m,l},\,\Pf(d_0,i^*)^{m,l+1}=\Pf(d_0,i+1^*)^{m,l}-\Pf(d_0,i^*)^{m,l}.
\end{align*}
\end{proposition}
The equation \eqref{2+1d} is still an $8$-point scheme, whose lattice points are the same with those given in \eqref{btoda}. Moreover, the discrete evolutions of the Pfaffian elements are almost the same with \eqref{de}. 
 These facts imply that there should be some unified relation between these eight vertices, and it may lead to different lattice equations when different interactions are imposed.

\section{Concluding remarks}
In this work, we mainly proposed two condensation algorithms for Pfaffians and relaxation factors are introduced in the algorithms as well. As we've shown, the computational cost of the algorithm is of $O(N^4)$ floating operations, and therefore,  the importance of the algorithms doesn't lie in themselves, but a better understanding of the interactions between discrete integrable systems and algebraic combinatorics.
The relaxation factors in the iteration processes \eqref{relax1} and \eqref{relax2} could not result in the $\lambda$-Pfaffian if we consider the exact value of $\tau_{2n}^{0,0}$. Therefore, how to correctly propose the concept of $\lambda$-Pfaffian and how to relate it to the combinatoric objects are still unknown to us; we leave it for our future study.

Another interesting question is about the reduction to $1$-dimensional lattice. For example, as the Toda lattice can be iterated and its solution can be explicitly expressed as determinant, if one consider the reduction $\tau_n^{k,l}:=\hat{\tau}_{n+2k+2l}$ where $\tau_n^{k,l}$ is the tau-function in \eqref{toda},
then $\{\hat{\tau}_\ell\}$ satisfy a Somos-4 sequence
\begin{align*}
\hat{\tau}_{\ell}\hat{\tau}_{\ell+4}=\hat{\tau}_{\ell+1}\hat{\tau}_{\ell+3}+\hat{\tau}_{\ell+2}^2.
\end{align*}
The exact Hankel determinant solution was given by \cite{chang12,hone19}. The solution of the Somos-5 sequence is then given by \cite{chang15} and the key idea is to make use of the B\"acklund transformation of the Somos-4 sequence. It can be made as a reduction from the fully discrete Lotka-Volterra lattice.
Although a sigma-function solution of the Somos-6 sequence was given in \cite{fedorov16}, it is still unknown for us about the explicit Pfaffian solution. The choice of $\tau_n^{k,l}=\hat{\tau}_{n+2k+4l}$ gives the reduction to the Somos-6 sequence
\begin{align*}
\hat{\tau}_{\ell+6}\hat{\tau}_\ell=\hat{\tau}_{\ell+5}\hat{\tau}_{\ell+1}-\hat{\tau}_{\ell+4}\hat{\tau}_{\ell+2}+\hat{\tau}_{\ell+3}^2.
\end{align*}
Moreover, for the generalised Lotka-Volterra lattice \eqref{glv}, if we set $\tau_n^{k,l}:=\hat{\tau}_{n+3k+5l-1}$, then the Somos-7 sequence
\begin{align*}
\hat{\tau}_{\ell+7}\hat{\tau}_\ell=\hat{\tau}_{\ell+6}\hat{\tau}_{\ell+1}+\hat{\tau}_{\ell+5}\hat{\tau}_{\ell+2}-\hat{\tau}_{\ell+4}\hat{\tau}_{\ell+3}
\end{align*}
will be obtained. It is our future work to obtain the explicit Pfaffian solution for these Somos-6 and Somos-7 sequences.

\section*{Acknowledgement}
The author would like to thank Prof. Andrew N. W. Hone, Xing-Biao Hu and Dr. Xiang-Ke Chang for helpful discussions and comments. ARC Centre of Excellence for Mathematical and Statistical Frontiers (ACEMS) is appreciated for the financial support.

\small
\bibliographystyle{abbrv}

\end{document}